\documentclass[10pt, letterpaper]{IEEEtran}
\usepackage{amsmath,amsfonts,amssymb}
\usepackage{graphicx}
\usepackage{siunitx}
\usepackage{booktabs}

\usepackage{mathrsfs}
\usepackage{amsthm}
\usepackage{algorithm}
\usepackage{algpseudocode}

\usepackage{subcaption}
\usepackage[T1]{fontenc}
\usepackage{microtype}
\microtypesetup{protrusion=true, expansion=true}

\newcommand{\Prb}{\mathbb{P}} 
\usepackage{tikz}
\usetikzlibrary{arrows.meta,positioning,shadows}

\usepackage{enumitem}

\usepackage{amsthm}

\theoremstyle{plain}
\newtheorem{theorem}{Theorem}
\newtheorem{lemma}{Lemma}
\newtheorem{corollary}{Corollary}

\theoremstyle{definition}
\newtheorem{definition}{Definition}
\newtheorem{assumption}{Assumption}

\theoremstyle{remark}

\usepackage{hyperref}

\usepackage{amsthm}

\usepackage{soul} 
\definecolor{ThananaEdit}{rgb}{0.0,0.3,0.9} 
\definecolor{ThananaMark}{rgb}{0.80,0.90,1.00} 

\title{Cognitive-Flexible Control via Latent Model Reorganization with Predictive Safety Guarantees}

\author{Thanana Nuchkrua and Sudchai Boonto, \IEEEmembership{Senior Member, IEEE}%
\thanks{Corresponding author: Sudchai Boonto.}
\thanks{The authors are with the Department of Control Systems and Instrumentation Engineering, King Mongkut's University of Technology Thonburi, Bangkok 10140, Thailand (e-mail: thanana.nuch@yahoo.com; sudchai.boo@kmutt.ac.th).}}

\begin{document}
	\maketitle
	
\begin{abstract}
Learning-enabled control systems must maintain safety when system dynamics and sensing conditions change abruptly. Although stochastic latent-state models enable uncertainty-aware control, most existing approaches rely on fixed internal representations and can degrade significantly under distributional shift. This letter proposes a \emph{cognitive-flexible control} framework in which latent belief representations adapt online, while the control law remains 	explicit and safety-certified. We introduce a Cognitive-Flexible Deep Stochastic State-Space Model (CF--DeepSSSM) that reorganizes latent representations subject to a bounded \emph{Cognitive Flexibility Index} (CFI), and embeds the adapted model within a 	Bayesian model predictive control (MPC) scheme. We establish guarantees on bounded posterior drift, recursive feasibility, and closed-loop stability. Simulation results under abrupt changes in system dynamics and observations demonstrate safe representation adaptation with rapid performance recovery, highlighting the benefits of learning-enabled, rather than learning-based, control for nonstationary cyber--physical systems.
\end{abstract}
	
\section{Introduction}
Learning-enabled control systems, \emph{i.e.}, cyber--physical systems
(CPSs), increasingly operate in physically interactive
environments where context shifts are unavoidable.
Changes in dynamics, sensing reliability, and interaction conditions can occur
abruptly, requiring controllers to remain safe and effective under evolving
latent behavior, especially in safety-critical
applications~\cite{Derler2012CPS}. 

A common response in learning-enabled control is to pair learned latent dynamics
models with constraint-aware predictive control, since model predictive control
(MPC) provides a principled mechanism for enforcing safety constraints under
uncertainty~\cite{Hewing2020LearningMPC}.
Within this paradigm, stochastic latent world models enable model-based learning
and control~\cite{Gedon2021DeepSSSM}.
Deep stochastic state-space models (Deep SSSMs), in particular, support belief
propagation and uncertainty-aware prediction through learned transition and
observation models~\cite{Fraccaro2017,Karl2017,Hafner2019}, while structured
priors and hybrid physics--learning formulations improve data
efficiency~\cite{Becker2021}.
However, most existing approaches treat the observation-to-latent mapping as
stationary and adapt primarily through parameter updates; under regime changes
or sensing variations, this can lead to representation mis-specification,
uncertainty miscalibration, and a loss of predictive safety.
Crucially, these latent world model frameworks provide limited mechanisms for
\emph{regulated representation reorganization} under distributional shift.

From a control perspective, the central challenge is therefore not only to learn
new parameters, but to determine \emph{when} internal latent representations
should be reorganized and \emph{how} such reorganization can be carried out
without violating safety during the transition.
Classical adaptive and robust control methods provide strong stability guarantees
under structured parametric uncertainty~\cite{Slotine1991,Ioannou1996}, but rely
on fixed model structures and do not accommodate changes in internal
representations.
More recent learning-based safe control approaches incorporate learned dynamics
and uncertainty into constraint-enforcing control laws, including robust and
adaptive MPC~\cite{Aswani2013}, predictive safety filters and
chance-constrained control~\cite{Wabersich2021PSF,Soloperto2023BayesianActuation},
and safe reinforcement learning methods based on Lyapunov conditions or
constrained policy optimization~\cite{Achiam2017,Chow2018,Berkenkamp2021SafeRL,
	Thananjeyan2021Safety}.
While these methods effectively regulate inputs under model uncertainty, they
typically assume a \emph{fixed internal representation}; under regime shifts,
this assumption can lead to miscalibrated uncertainty, overly conservative
behavior, or loss of safety guarantees.

In parallel, cognitive flexibility has been studied as the ability to
adapt internal representations in response to changing contexts~\cite{Scott1962}.
Related ideas appear in meta-learning and rapid adaptation frameworks, where
representations or update rules are adjusted online to improve performance under
distributional shift~\cite{BelmonteBaeza2022MetaRL,McClement2022MetaRL,
	Goldie2024LearnedOptRL,Goldie2025MetaRL}.
However, these approaches are largely performance-driven and do not address how
latent representation changes should be regulated to preserve safety, a
limitation that is particularly critical in learning-enabled control where
representation changes directly affect uncertainty calibration.

Motivated by this gap, this letter introduces a cognitive-flexible control
framework that enables online reorganization of latent belief models while
maintaining predictive safety.
Representation adaptation is explicitly regulated and coupled with adaptive
constraint tightening, allowing the controller to respond to distributional
shifts without violating safety guarantees during transition.

\textbf{Contributions.}
This letter makes the following contributions.
\textbf{(i)} We formalize \emph{cognitive flexibility} in stochastic control as     
the regulated reorganization of latent belief representations, going beyond classical adaptive and robust control
frameworks that assume fixed model structures~\cite{Ioannou1996}.
\textbf{(ii)} We propose a cognitive-flexible Deep Stochastic State-Space Model
(CF--DeepSSSM) that enables online posterior restructuring, unlike existing
latent world models~\cite{Goldie2025MetaRL} that adapt only through parameter updates under stationary
representations~\cite{Hafner2019,Becker2021,Gedon2021DeepSSSM}.
\textbf{(iii)} We develop a safety-certified control mechanism with adaptive
uncertainty tightening that preserves constraint satisfaction during model
evolution, complementing prior safe and learning-based MPC approaches that
assume fixed internal representations~\cite{Aswani2013,Hewing2020LearningMPC,Wabersich2021PSF}.
\textbf{(iv)} We establish theoretical guarantees of bounded posterior drift and
closed-loop stability, extending existing safety and stability results for
learning-enabled control~\cite{Berkenkamp2021SafeRL}, and
validate the proposed approach in simulation under abrupt dynamics and
observation shifts.

The remainder of this letter is organized as follows.
Section~\ref{sec:Formulation} formulates the problem and introduces the modeling
assumptions.
Section~\ref{sec:method} presents the proposed CF--DeepSSSM control architecture.
Section~\ref{sec:theory} establishes theoretical guarantees on bounded posterior
drift, recursive feasibility, and closed-loop stability.
Simulation results are reported in Section~\ref{sec:simulation}, followed by
concluding remarks in Section~\ref{sec:conclusion}.
	
	\section{Preliminaries and Problem Formulation}
\label{sec:Formulation}

\subsection{Preliminaries}

We consider a partially observable stochastic dynamical system
\begin{equation}
	\label{eq:plant_pomdp} 
	x_{t+1} = f(x_t,u_t,w_t), 
	\qquad
	o_t = h(x_t,v_t),
\end{equation}
where $x_t \in \mathbb{R}^n$ denotes the (unobserved) physical interaction state,
$u_t \in \mathbb{R}^m$ the control input, and
$o_t \in \mathbb{R}^p$ the measured observation.
The disturbances $w_t$ and $v_t$ have unknown and possibly time-varying
distributions. Since $x_t$ is not directly measurable, control decisions
must rely on inferred internal representations.
Safety is specified in the physical state-input space by the admissible set
\begin{equation}
	\label{eq:safety_set}  
		(x_t,u_t) \in \mathcal{S}
		:= \{(x,u) \mid \mathcal{G}_i(x,u) \le 0,\; i=1,\dots,q\},
	\end{equation}
	where each $\mathcal{G}_i$ is Lipschitz and encodes biomechanical and
	actuation limits. The safety requirement is probabilistic:
	\begin{equation}
		\label{eq:prob_safety}
		\mathbb{P}\big((x_t,u_t)\in\mathcal{S}\big) \ge 1-\delta,
	\end{equation}
	for a prescribed violation level $\delta\in(0,1)$.
	Because the state $x_t$ is unobserved, constraint satisfaction must be
	verified through a predictive distribution over physical states.
	To this end, the controller maintains a compact latent belief
		$z_t \in \mathbb{R}^k$, $k\ll n$, inferred from the interaction history
		$\mathcal{H}_t=\{o_0,u_0,\dots,o_t\}$ via
	\begin{equation}
		\label{eq:latent_belief}                
		z_t \sim q_{\phi_t}(z_t \mid \mathcal{H}_t),
	\end{equation}
	where $\phi_t$ parameterizes the inference mapping.
	The pair $(z_t,\phi_t)$ summarizes the posterior belief over the physical state.
	
	Belief and physical space are connected through a predictive decoder
	\begin{equation}
		\label{eq:decoder}
		x_t \sim p_\theta(x_t \mid z_t),
	\end{equation}
	induced by a stochastic latent dynamics model
	\begin{equation}
		\label{eq:deepsssm} 
		z_{t+1} \sim p_\theta(z_{t+1}\mid z_t,u_t),
		\qquad
		o_t \sim p_\theta(o_t\mid z_t).
	\end{equation}
	The parameter $\theta$ defines both the latent transition and the
	latent-to-physical predictive mapping. The resulting belief-induced
		predictive distribution over physical states is
	\begin{equation}
		\label{eq:predictive_physical}
		p_\theta(x_t \mid \mathcal{H}_t)
		=
		\int p_\theta(x_t \mid z_t)\,
		q_{\phi_t}(z_t \mid \mathcal{H}_t)\,dz_t .
	\end{equation}
	
	Safety constraints in \eqref{eq:safety_set} are evaluated with respect to
		the predictive physical distribution defined in \eqref{eq:predictive_physical}.
	In particular, both the expectation in \eqref{eq:control_policy} and the
		probabilistic safety requirement \eqref{eq:prob_safety} are defined with
		respect to the same belief-induced measure over $x_t$, ensuring that
		performance evaluation and constraint enforcement are probabilistically
		consistent.

	
	\subsection{Problem Formulation}
	Given this predictive model, the control policy $\pi$ generates
		$u_t=\pi(z_t)$ and solves
	\begin{equation}
		\label{eq:control_policy}  
		\min_{\pi}
		\;
		\mathbb{E}_{x_t \sim p_\theta(\cdot\mid \mathcal{H}_t)}
		\!\left[
		\sum_{t=0}^{T}\ell(x_t,u_t)
		\right]
	\end{equation}
	subject to the probabilistic safety requirement \eqref{eq:prob_safety},
	where $\ell(x_t,u_t)$ is a stage cost defined in physical space.
	Thus, control is computed in latent space, while performance and safety
	are defined in physical space through the predictive mapping
	\eqref{eq:predictive_physical}.

	
	Under distributional shift, the inference mapping must adapt.
	We define cognitive flexibility (CF) as a regulated evolution
	\begin{equation}
		\label{eq:cf_constraint} 
		\|\phi_t-\phi_{t-1}\| \le \epsilon,
	\end{equation}
	where $\epsilon>0$ bounds the instantaneous rate of latent
	reorganization and prevents excessive representation drift.
	
	\textbf{Problem Statement:}
	Design a latent-state feedback policy $u_t=\pi(z_t)$ such that
		(i) the probabilistic safety requirement
		$\mathbb{P}((x_t,u_t)\in\mathcal{S})\ge 1-\delta$
		holds for all $t$,
		(ii) the expected physical cost \eqref{eq:control_policy} is minimized,
		(iii) the inference evolution satisfies \eqref{eq:cf_constraint}.
	The core challenge is that safety is defined in physical space
	through the constraints in \eqref{eq:safety_set}--\eqref{eq:prob_safety},
  whereas control is computed over a dynamically evolving latent belief
	inferred in \eqref{eq:latent_belief}


\section{Proposed CF--DeepSSSM Method}
\label{sec:method}
To address the problem in Sec.~\ref{sec:Formulation}, we propose
\emph{Cognitive-Flexible DeepSSSM} (CF--DeepSSSM), a closed-loop
architecture that combines latent-state modeling, belief-space
predictive safety control, and surprise-regulated representation
adaptation.	
We first model the system dynamics in
\eqref{eq:plant_pomdp} through a Deep SSSM \footnote{Deep SSSM is “deep” in the sense that the inference encoder 
		$q_{\phi_t}(z_t\mid\mathcal{H}_t)$, in~\eqref{eq:latent_belief} the latent transition 
		$p_{\theta_t}(z_{t+1}\mid z_t,u_t)$ in~\eqref{eq:deepsssm}, and the observation model 
		$p_{\theta_t}(o_t\mid z_t)$ in~\eqref{eq:deepsssm} are parameterized by multi-layer 
		perceptrons (two hidden layers with ReLU activations), producing 
		both predictive mean and covariance in latent space.}     defined by \eqref{eq:deepsssm}.  
The evolution of the latent belief is described by
\begin{equation}
	p_{\theta_t}(z_{t+1}\mid z_t,u_t), \qquad
	p_{\theta_t}(o_t\mid z_t),
	\label{eq:latentbelief}
\end{equation}
where the model parameters $\theta_t$ are learned via stochastic variational
inference.
This formulation yields a compact latent representation together with a
	predictive covariance $\Sigma_t := \Sigma_{\theta_t}(z_t,u_t)$,
	where $\Sigma_{\theta_t}$ denotes the latent process noise covariance.
	We do not assume perfect Bayesian calibration; instead,
	$\Sigma_t$ is used as a conservative uncertainty proxy
	for safety reasoning and constraint tightening.
	In particular, safety guarantees rely only on a bounded
	one-step prediction error, rather than exact probabilistic calibration.
The resulting uncertainty captures modeling error induced by partial
observability and evolving interaction conditions, and serves as the primary
signal for safety-aware decision making with respect to the constraints defined
in \eqref{eq:safety_set}.

Given this probabilistic latent dynamics model, safety can be enforced by
planning directly over the predictive belief distribution.
This naturally leads to a predictive control formulation, instantiated here as
Bayesian Model Predictive Control (BMPC). 
Safety is enforced through a BMPC layer operating on the latent belief
\eqref{eq:latent_belief}.
{At each time step, the controller formulated in~\eqref{eq:control_policy}
	computes a horizon-$T$ control sequence by solving
	\begin{equation}
		\min_{u_{t:t+T-1}}
		\mathbb{E}\!\sum_{k=0}^{T-1}\ell(z_{t+k},u_{t+k})
		\;\;\text{s.t.}\;\;
		\Prb\!\left(\mathcal{G}(z_{t+k},u_{t+k})\!\le\!0\right)\!\ge\!1-\epsilon,
		\label{eq:horizon}
	\end{equation}
	where $\ell(\cdot)$ encodes tracking and comfort objectives.
	The probability constraint is evaluated using the predictive belief
		distribution obtained by propagating the current latent belief
	\eqref{eq:latent_belief} through the Deep SSSM dynamics~\eqref{eq:latentbelief}.
	Over the prediction horizon, the latent predictive mean and covariance
		$(\hat z_{t+k}, \Sigma_{t+k})$ are propagated recursively via the stochastic
		latent transition model. Specifically, at each prediction step $k$, the mean
		$\hat z_{t+k}$ evolves according to the predictive mean dynamics, while the
		stage-wise covariance $\Sigma_{t+k}=\Sigma_{\theta_t}(\hat z_{t+k},u_{t+k})$
		is obtained from the latent process-noise model.
	While predictive control, i.e., BMPC, governs how control inputs
	$u_t$ are selected safely—such that the safety constraints
	$(x_t, u_t) \in \mathcal{S}$ in \eqref{eq:safety_set} are satisfied with high
	probability—it does not by itself indicate when the underlying latent dynamics
	should be revised.
	To monitor the validity of the learned model during ongoing interaction, we
	introduce an instantaneous measure of \emph{surprise},
	\begin{equation}
		\mathcal{S}_t
		:= -\log p_{\theta_t}\!\left(o_{t+1}\mid z_t,u_t\right),
		\label{eq:surprise}
	\end{equation}
	which quantifies discrepancies between predicted and observed outcomes.
	After applying $u_t$ and observing $o_{t+1}$, the model parameters $\theta_t$ are updated via
	\begin{equation}
		\theta_{t+1}
		= \theta_t + \eta_t \nabla_\theta
		\log p_{\theta_t}\!\left(o_{t+1}\mid z_t,u_t\right),
		\label{eq:param_update}
	\end{equation}
	where the adaptation rate is modulated by $\mathcal{S}_t$.
	Large surprise values induce faster adaptation, while diminishing step sizes
	$\eta_t$ satisfying
	$0<\eta_t\le\eta_{\max}$,
	$\sum_t \eta_t=\infty$,
	and $\sum_t \eta_t^2<\infty$
	ensure bounded parameter drift and long-term stability.   
	Thus, representation refinement follows \eqref{eq:param_update}, 
	while safety is enforced by BMPC~\eqref{eq:horizon}.
	To ensure controlled reorganization of the latent belief, adaptation is
	explicitly regulated through the CF constraint
	\begin{equation}
		\mathbb{E}\!\left[\|\phi_{\theta_{t+1}}-\phi_{\theta_t}\|\right] \le \epsilon,
		\label{eq:cognitive_flexibility}
	\end{equation}
	which bounds the rate of change of the inference mapping and preserves predictive
	safety during online adaptation. 
	The CF constraint bounds the evolution of the inference mapping
	$\phi_{\theta_t}$ within the fixed Deep SSSM model structure
	(encoder, transition, and observation parameterization),
	thereby regulating the geometry of the latent belief space.
	Thus, the controller \eqref{eq:horizon} can respond to
	changes in interaction conditions—detected via elevated surprise
	\eqref{eq:surprise}—while preserving predictive safety for
	safety-critical operation.
	This constraint~\eqref{eq:cognitive_flexibility} is enforced through the
		surprise-driven parameter update~\eqref{eq:param_update}, where the
		step size $\eta_t$ ensures bounded parameter drift.
	

	\textbf{CF Index (CFI).}
	To quantify the instantaneous magnitude of representation reorganization,
	we define the empirical CFI 
	\begin{equation}
		\mathrm{CFI}_t := \frac{\|\phi_{\theta_{t+1}} - \phi_{\theta_t}\|}{\epsilon}.
		\label{eq:cfi}
	\end{equation}
	Under the expectation-based cognitive-flexibility constraint
	$\mathbb{E}[\|\phi_{\theta_{t+1}} - \phi_{\theta_t}\|] \le \epsilon$,
	we have $\mathbb{E}[\mathrm{CFI}_t] \le 1$.
	
	Overall, CF--DeepSSSM forms a closed-loop architecture that
	integrates latent-state inference, uncertainty-aware BMPC,
	and surprise-regulated adaptation. The resulting control
	pipeline and update loop are summarized in Fig.~\ref{fig:cfsssm_overview}
	and Algorithm~\ref{alg:cfdsssm}.

	\begin{figure}[t]
		\centering
		\resizebox{1.0\linewidth}{!}{%
				\begin{tikzpicture}[
					font=\small\sffamily,
					block/.style={
						draw,
						rounded corners=8pt,
						minimum width=3.8cm,
						minimum height=0.85cm,
						align=center,
						line width=0.6pt,
						drop shadow={opacity=0.08}
					},
					line/.style={-Latex, line width=0.9pt},
					node distance=2.8cm
					]
					\node[block, fill=gray!8] (obs)
					{\textbf{Sensing$/$Measurement}\\ \Large $o_t$ \eqref{eq:deepsssm}};
					
					\node[block, fill=blue!7, right=of obs] (latent)
					{\textbf{Deep SSSM Inference}\\ [3 pt]
						\large $z_t,\;\Sigma_t$ \eqref{eq:latentbelief}};
					
					\node[block, fill=orange!12, right=of latent] (mpc)
					{\textbf{Predictive Safety Control}\\ BMPC, \;\large $u_t$ \eqref{eq:horizon}-\eqref{eq:surprise}};
					
					
					\node[block, fill=gray!15, right=of mpc] (plant)
					{\textbf {Physical System}\\[3 pt]
						\normalsize $(\mathscr{A}_{\mathrm{env}},\,\mathscr{B}_{\mathrm{env}},\,\mathscr{C}_{\mathrm{env}})$ (sec.\ref{sec:simulation})};
					
					\node[block, fill=green!8, below=0.76cm of latent] (adapt)
					{\textbf{Cognitive Adaptation}\\ \large $\theta_t,\;\pi_t$ \eqref{eq:param_update}};
					
					\draw[line] (obs) -- node[above]{encode} (latent);
					\draw[line] (latent) -- node[above]{belief} (mpc);
					\draw[line] (mpc) -- node[above]{control} (plant);
					
					\draw[line] (plant.north) |- ++(-1.1,0.8) -|
					node[left,pos=0.65]{observations} (obs.north);
					
					\draw[line] (latent.south) -- node[left]{prediction error} (adapt.north);
					\draw[line] (adapt.east) -|
					node[below,pos=0.25]{bounded update} (mpc.south);
					\draw[line] (adapt.north) -| (latent.south);

				\end{tikzpicture}%
			}
			\caption{System overview of CF--DeepSSSM.} 
	\label{fig:cfsssm_overview}
\end{figure}

\begin{algorithm}[t]  
	\scriptsize
	\caption{\footnotesize CF--DeepSSSM Predictive Safety Control (Sec.~\ref{sec:method})\protect\footnotemark}
	\label{alg:cfdsssm}
	\begin{algorithmic}[1]
		\State Initialize belief $z_0$, parameters $\theta_0$
		\For{$t=0,1,2,\dots$}
		\State Infer latent belief
		$z_t \sim q_{\phi_t}(z_t \mid \mathcal{H}_t)$
		\hfill\textit{~\eqref{eq:deepsssm}}
		
		\State Compute safe control
		$u_t \leftarrow \mathrm{BMPC}(z_t,\Sigma_{z,t})$
		\hfill\textit{~\eqref{eq:horizon}}
		
		\State Observe $o_{t+1}$ and compute surprise
		$\mathcal{S}_t$ 
		\hfill\textit{~\eqref{eq:surprise}}
		
		\State Update model parameters
		$\theta_{t+1} $ 
		\hfill\textit{~\eqref{eq:param_update}}
		\EndFor
	\end{algorithmic}
\end{algorithm}
\footnotetext{Each iteration requires a forward inference pass,
		surprise evaluation, a single bounded parameter update,
		and solution of a convex quadratic BMPC problem.
		For state dimension $n$ and horizon $N$, a dense QP solve
		scales as $\mathcal{O}((nN)^3)$ in the worst case,
		with reductions from structure and sparsity.}

\section{Theoretical Foundations of CF--DeepSSSM} 
\label{sec:theory}


We analyze the closed-loop properties of the CF--DeepSSSM controller
introduced in Sec.~\ref{sec:method}. In particular, representation
reorganization is regulated by the CF constraint
\eqref{eq:cognitive_flexibility}, predictive safety is enforced through
belief-space BMPC~\eqref{eq:horizon}, and model adaptation is driven by
the surprise signal~\eqref{eq:surprise}. 

\paragraph{Cognitive-flexible latent dynamics (abstract analysis).}
This abstraction captures the effect of the surprise-driven updates in
\eqref{eq:surprise}--\eqref{eq:param_update} applied to the Deep SSSM model
\eqref{eq:deepsssm}, while explicitly separating latent state $z_{t+1}$ and model parameter $\theta_{t+1}$ evolutions from
representation evolution for theoretical analysis is formalized by the following dynamics:
\begin{equation}
z_{t+1} = f_{\theta_t}(z_t,u_t) + w_t, \qquad
\theta_{t+1} = \theta_t + \alpha_t \Delta_t .
\label{eq:flexible-latent}
\end{equation}
Here, $f_{\theta_t}$ denotes the predictive mean induced by the latent dynamics,
$\alpha_t \ge 0$ is a (possibly time-varying) step size, and $\Delta_t$ is a
bounded update direction driven by predictive surprise.

\begin{definition}[Bounded posterior drift]
The latent model update $\theta_{t+1}$ in \eqref{eq:flexible-latent} is stated to satisfy
\emph{cognitive regularity} if
$
\|\theta_{t+1}-\theta_t\| \le \rho(\mathcal{S}_t),
$
where 
$\rho(\cdot)$ is a nondecreasing function.
This condition ensures that representation reorganization is
data-justified and rate-limited.
\end{definition}

\paragraph{Belief uncertainty model.}
Consistent with the stochastic latent modeling introduced in~\eqref{eq:latentbelief}, belief evolution is represented by a probabilistic latent
dynamics model 
$
p_{\theta_t}(z_{t+1}\mid z_t,u_t)
= \mathcal{N}\!\big(f_{\theta_t}(z_t,u_t),
\Sigma_{\theta_t}(z_t,u_t)\big),
$
with the latent-space observation model $p_{\theta_t}(o_t\mid z_t)$ defined in \eqref{eq:deepsssm}.
Here, $f_{\theta_t}$ denotes the predictive mean parameterized by $\theta_t$. 
For analysis, we assume that online inference maintains a variational
factorization
$
p(z_t,\theta_t \mid o_{1:t})
\approx q_{\phi_t}(z_t)\,q_{\psi_t}(\theta_t),
$
where $q_{\phi_t}(z_t)$ denotes the variational posterior over $z_t$ and $q_{\psi_t}(\theta_t)$ denotes a variational belief over $\theta_t$. 
This mean-field approximation yields calibrated predictive uncertainty used
for safety reasoning.

\paragraph{Predictive safety mechanism.}
The BMPC policy introduced in Sec.~\ref{sec:method} enforces safety by
planning over the latent belief dynamics while respecting the state--input
constraints defined in Sec.~\ref{sec:Formulation}.
To account for modeling error arising from partial observability and ongoing
latent model adaptation, constraint satisfaction in \eqref{eq:horizon} is enforced
through adaptive tightening.
Specifically, each constraint $\mathcal{G}_i(\cdot)$ in \eqref{eq:safety_set} is
modified as $\mathcal{G}_i(z,u) \le -\beta_{i,t}$, where the tightening margin
$\beta_{i,t} = c_i\,\mathcal{S}_t$ scales with the predictive surprise
$\mathcal{S}_t$ in \eqref{eq:surprise}, and $c_i>0$ denotes a
constraint-specific sensitivity coefficient. 
Together, \eqref{eq:flexible-latent} and the predictive safety mechanism ensure
recursive feasibility of the belief-space control \eqref{eq:horizon} under
bounded adaptation.

\begin{assumption}[Model and safety regularity]
\label{ass:ms}
The latent dynamics $f_\theta(z,u)$ are Lipschitz in $(z,u)$ $\forall{\theta}$,
the process noise has bounded second moment, and the initial belief has bounded
support (or variance). The admissible set
$\mathcal{S}=\{(z,u): \mathcal{G}_i(z,u)\le0\}$ is compact (convex when required), and each
$\mathcal{G}_i$ is Lipschitz continuous.
\end{assumption}

\begin{assumption}[Incremental adaptation]
\label{ass:adapt}
Model updates are incremental, rewards are bounded, and latent estimation error
remains uniformly bounded during adaptation.
\end{assumption}

The following result shows that surprise-regulated adaptation
\eqref{eq:surprise}--\eqref{eq:param_update} bounds latent model reorganization
\eqref{eq:flexible-latent}, which is necessary to preserve predictive safety
under belief-space control \eqref{eq:horizon}.

\begin{theorem}[Bounded posterior drift]
\label{thm:drift}
Assume the update direction $\Delta_t$ is uniformly bounded,
$\|\Delta_t\|\le L_\Delta$ almost surely, and the adaptation rate satisfies
$\alpha_t \le \frac{\eta}{1+\mathcal{S}_t}$ with $\mathcal{S}_t\ge 0$.
Then, 
$
\|\theta_{t+1}-\theta_t\|\le \eta L_\Delta$, with $\eta>0$ is a design constant.  
\end{theorem} 
\begin{proof}
From the update \eqref{eq:flexible-latent} we have
$
\theta_{t+1}-\theta_t=\alpha_t\Delta_t
\quad\Rightarrow\quad
\|\theta_{t+1}-\theta_t\|=\|\alpha_t\Delta_t\|
\le \alpha_t\|\Delta_t\|.
$
By Assumption \ref{ass:adapt}, $\|\Delta_t\|\le L_\Delta$ a.s. and
$\alpha_t \le \frac{\eta}{1+\mathcal{S}_t}$ with $\mathcal{S}_t\ge 0$, hence
$
\|\theta_{t+1}-\theta_t\|
\le \frac{\eta}{1+\mathcal{S}_t}L_\Delta
\le \eta L_\Delta,
$
since $(1+\mathcal{S}_t)^{-1}\le 1$, $\forall{\mathcal{S}_t}\ge 0$.
Therefore $\|\theta_{t+1}-\theta_t\|\le \eta L_\Delta$, $\forall{t}$. \qedhere
\end{proof}

\begin{lemma}[Tightening dominates prediction mismatch]
\label{lem:tightening}
Suppose $\mathcal{G}_i(z,u)$ is $L_{g,i}$-Lipschitz in $z$, and the Deep SSSM
predictive distribution satisfies
$z_{t+1}\sim\mathcal{N}(\hat z_{t+1},\Sigma_t)$ with
$\sigma_t=\sqrt{\lambda_{\max}(\Sigma_t)}$.
If $\beta_{i,t} \ge L_{g,i}\sigma_t$, then
$\mathcal{G}_i(\hat z_{t+1},u_t)\le -\beta_{i,t}$
implies $\mathcal{G}_i(z_{t+1},u_t)\le 0$
with probability at least $1-\delta_i$, and $\delta_i\in(0,1)$ denote the allowable violation probability of constraint $i$.
\end{lemma}
\begin{proof}
Fix any constraint $i$ and time $t$. By Lipschitz continuity of
$\mathcal{G}_i$ and the one-step prediction error bound,
$
\mathcal{G}_i(z_{t+1},u_t)
\le \mathcal{G}_i(\hat z_{t+1},u_t) + L_{g,i}\sigma_t .
$
If the tightened constraint satisfies
$\mathcal{G}_i(\hat z_{t+1},u_t)\le -\beta_{i,t}$ with
$\beta_{i,t}\ge L_{g,i}\sigma_t$, then
$\mathcal{G}_i(z_{t+1},u_t)\le 0$.
Thus, whenever the prediction error bound holds, feasibility of the
tightened constraint implies feasibility of the true constraint.
Since the bound holds with probability at least $1-\delta_i$ under the
Deep SSSM predictive distribution, we obtain
$\Prb(\mathcal{G}_i(z_{t+1},u_t)\le 0)\ge 1-\delta_i$.
\end{proof}

From Lemma~\ref{lem:tightening}, the predictive-safety mechanism can be stated explicitly as follows.
The Deep SSSM yields a predictive covariance $\Sigma_t$,
which induces a one-step belief error bound
$\sigma_t = \sqrt{\lambda_{\max}(\Sigma_t)}$.
The tightening margin $\beta_{i,t}$ is selected as a function of this
uncertainty bound.
Accordingly, feasibility of the tightened constraint
$\mathcal{G}_i(\hat z_{t+1},u_t)\le -\beta_{i,t}$
implies feasibility of the original constraint
$\mathcal{G}_i(z_{t+1},u_t)\le 0$
with probability at least $1-\delta_i$.
Hence, predictive uncertainty is systematically translated into
a probabilistic safety guarantee.


\begin{theorem}[Recursive feasibility]
\label{thm:feasibility}
Let $\mathcal{Z}_f$ be a terminal set and
$\kappa_f:\mathcal{Z}_f\!\to\!\mathbb{R}^m$ satisfy
$(z,\kappa_f(z))\in\mathcal{S}$, $\forall z\in\mathcal{Z}_f$.
Consider~\eqref{eq:horizon} with tightened constraints
$G_i(\hat z_{t+k|t},u_{t+k|t})\le-\beta_{i,t+k}$,
$k=0,\dots,T-1$, where $\beta_{i,t+k}$ follows Lemma~\ref{lem:tightening}
under Assumption~\ref{ass:ms}.
If~\eqref{eq:horizon} is feasible at time $t$, then it is feasible at $t+1$,
hence $\forall t'\ge t$.
\end{theorem}

\begin{proof}
	Let $\mathcal Z_f$ denote a terminal invariant set and
	$\kappa_f:\mathcal Z_f\rightarrow\mathbb R^m$ a terminal admissible control
	law such that $(z,\kappa_f(z))\in\mathcal S$ for all $z\in\mathcal Z_f$.
	Consider the predictive control problem~\eqref{eq:horizon} defined on the belief dynamics
	$
	z_{t+1} 
	$ in~\eqref{eq:flexible-latent}
  with tightened constraints $G_i(z_{t+k},u_{t+k})\le-\beta_{i,t+k}$ obtained
	from Lemma~\ref{lem:tightening}.
  Assume the problem is feasible at time $t$ and let
	$
	\mathbf u_t^\star=\{u_{t|t}^\star,\dots,u_{t+T-1|t}^\star\}
	$
	be an optimal feasible control sequence with predicted states
	$\{\hat z_{t+k|t}^\star\}_{k=1}^T$ satisfying the tightened constraints and
	$\hat z_{t+T|t}^\star\in\mathcal Z_f$.
	Applying the standard MPC shift argument, consider the candidate sequence
	at time $t+1$
	$
	\tilde{\mathbf u}_{t+1}
	=
	\{u_{t+1|t}^\star,\dots,u_{t+T-1|t}^\star,\kappa_f(\hat z_{t+T|t}^\star)\}.
	$
	Under Assumption~\ref{ass:ms}, the one-step prediction error induced by the
	process noise $w_t$ and the model mismatch of the belief dynamics is bounded.
	By Lemma~\ref{lem:tightening}, the tightening margin $\beta_{i,t}$ is chosen
	so that this prediction mismatch is dominated by the tightening bound.
	Hence whenever the tightened constraint
	$
	G_i(\hat z_{t+1|t},u_{t|t})\le-\beta_{i,t}
	$
	holds, the true constraint
	$
	G_i(z_{t+1},u_t)\le0
	$
	is satisfied with probability at least $1-\delta_i$.
	Since the shifted sequence preserves feasibility of the tightened problem and
	the terminal condition $\hat z_{t+T|t}^\star\in\mathcal Z_f$ together with the
	terminal control $\kappa_f$ guarantees admissibility at the terminal step,
	the candidate sequence is feasible at time $t+1$.
	Therefore feasibility at time $t$ implies feasibility at time $t+1$, which
	establishes recursive feasibility.
	\end{proof}
	

	\begin{theorem}[ISS under cognitive-flexible adaptation]
\label{thm:iss}
Under Assumptions~\ref{ass:ms}--\ref{ass:adapt}, the closed-loop belief
dynamics can be written as
$z_{t+1}=f_{\theta_t}(z_t,u_t)+d_t$, where $d_t$ is a bounded disturbance
($\|d_t\|\le \bar d$) aggregating the effects of process noise,
model mismatch, and bounded parameter variation.
Then there exist a class-$\mathcal{KL}$ function $\beta$ and a
class-$\mathcal{K}$ function $\gamma$ such that for all $t\ge0$,
$\|z_t\|\le \beta(\|z_0\|,t)+\gamma(\bar d)$.
Equivalently, the MPC value function $V_t(\cdot)$ satisfies the ISS decrease
condition
$V_{t+1}(z_{t+1})-V_t(z_t)\le -\alpha(\|z_t\|)+\gamma(\|d_t\|)$
for some class-$\mathcal{K}$ functions $\alpha,\gamma$.
\end{theorem}
\begin{proof}
Under the terminal ingredients required for recursive feasibility
	(Theorem~\ref{thm:feasibility}), the optimal value function of the MPC
	problem~\eqref{eq:horizon},
  $
	V_t(z_t)=\min_{u_{t:t+T-1}}
	\mathbb E\!\left[\sum_{k=0}^{T-1}\ell(z_{t+k},u_{t+k})\right],
	$
	acts as a Lyapunov function for the nominal belief dynamics
	$z_{t+1}=f_{\theta_t}(z_t,u_t)$ in~\eqref{eq:flexible-latent}.
   When the predictive model~\eqref{eq:control_policy} is imperfect and updated online, the closed-loop
	belief evolution follows
	$
	z_{t+1}=f_{\theta_t}(z_t,u_t)+w_t ,
	$
	where the disturbance term $w_t$ captures both process noise and prediction
	mismatch induced by model uncertainty.
	Moreover, the parameter update
	$
	\theta_{t+1}=\theta_t+\alpha_t\Delta_t
	$
	introduces a time-varying perturbation to the dynamics.
	By Theorem~\ref{thm:drift}, the parameter variation satisfies
	$
	\|\theta_{t+1}-\theta_t\|\le \eta L_\Delta ,
	$
	which implies that the induced modeling error in
	$f_{\theta_t}(z_t,u_t)$ remains uniformly bounded under
	Assumption~\ref{ass:ms}.
   Hence the closed-loop belief dynamics can be written as the nominal
	system plus a bounded disturbance
	$
	z_{t+1}=f_{\theta_t}(z_t,u_t)+d_t ,
	\qquad \|d_t\|\le \bar d .
	$
	Using the standard ISS-MPC argument, the value function decrease
	condition holds in the form
	$
	V_{t+1}(z_{t+1})-V_t(z_t)
	\le
	-\alpha(\|z_t\|)+\gamma(\|d_t\|),
	$
	for class-$\mathcal K$ functions $\alpha,\gamma$.
	Since $d_t$ is bounded, the inequality establishes that the closed-loop
	belief dynamics are input-to-state stable with respect to bounded modeling
	error and bounded parameter drift induced by cognitive-flexible adaptation.
	\end{proof}


\begin{corollary}[Safety preservation]
\label{cor:safety}
Under Assumptions~\ref{ass:ms} and Theorem~\ref{thm:feasibility}, the
closed-loop sequence $\{(z_t,u_t)\}_{t\ge0}$ generated by the BMPC law
satisfies
$
(z_t,u_t)\in\mathcal{S}_z,\qquad \forall t\ge0,
$
where $\mathcal{S}_z := \{(z,u)\mid G_i(z,u)\le0,\ i=1,\dots,q\}$.
Equivalently, the stage-wise tightened constraints in~\eqref{eq:horizon}
hold for all prediction steps $k=0,\dots,T-1$ along the implemented closed loop.
Moreover, by Lemma~\ref{lem:tightening}, satisfaction of the tightened
constraints implies the physical chance constraint
$\eqref{eq:prob_safety}$, $\forall t\ge0$ with violation level $\delta$.
\end{corollary}

\begin{proof}
By Theorem~\ref{thm:feasibility}, the tightened MPC problem~\eqref{eq:horizon} is recursively feasible.
Hence, for each time $t$ there exists a feasible control sequence
$\mathbf u_t^\star := \{u_{t|t}^\star,\dots,u_{t+T-1|t}^\star\}$
with predicted states $\{z_{t+k|t}^\star\}_{k=0}^{T}$ such that the stage-wise tightened constraints
$G_i(z_{t+k|t}^\star,u_{t+k|t}^\star)\le -\beta_{i,t+k}$ hold for all $k=0,\dots,T-1$ along the predicted trajectory.
The controller applies the first element $u_t := u_{t|t}^\star$, so in particular the $k=0$ constraint satisfies
$G_i(z_{t|t}^\star,u_{t|t}^\star)\le -\beta_{i,t}$.
By Lemma~\ref{lem:tightening}, satisfaction of the tightened constraint at time $t$ implies that the physical chance constraint holds at time $t$, i.e.,
$\mathbb{P}\big(G_i(x_t,u_t)\le 0\big)\ge 1-\delta_i$ for each constraint $i$.
Equivalently, with $\delta := \max_i \delta_i$, we obtain
$\mathbb{P}\big((x_t,u_t)\in\mathcal S\big)\ge 1-\delta$, $\forall{t}$.
\end{proof}

\begin{figure*}[t]
	\centering
	\begin{subfigure}[t]{0.16\textwidth}
		\centering
		\includegraphics[width=\linewidth]{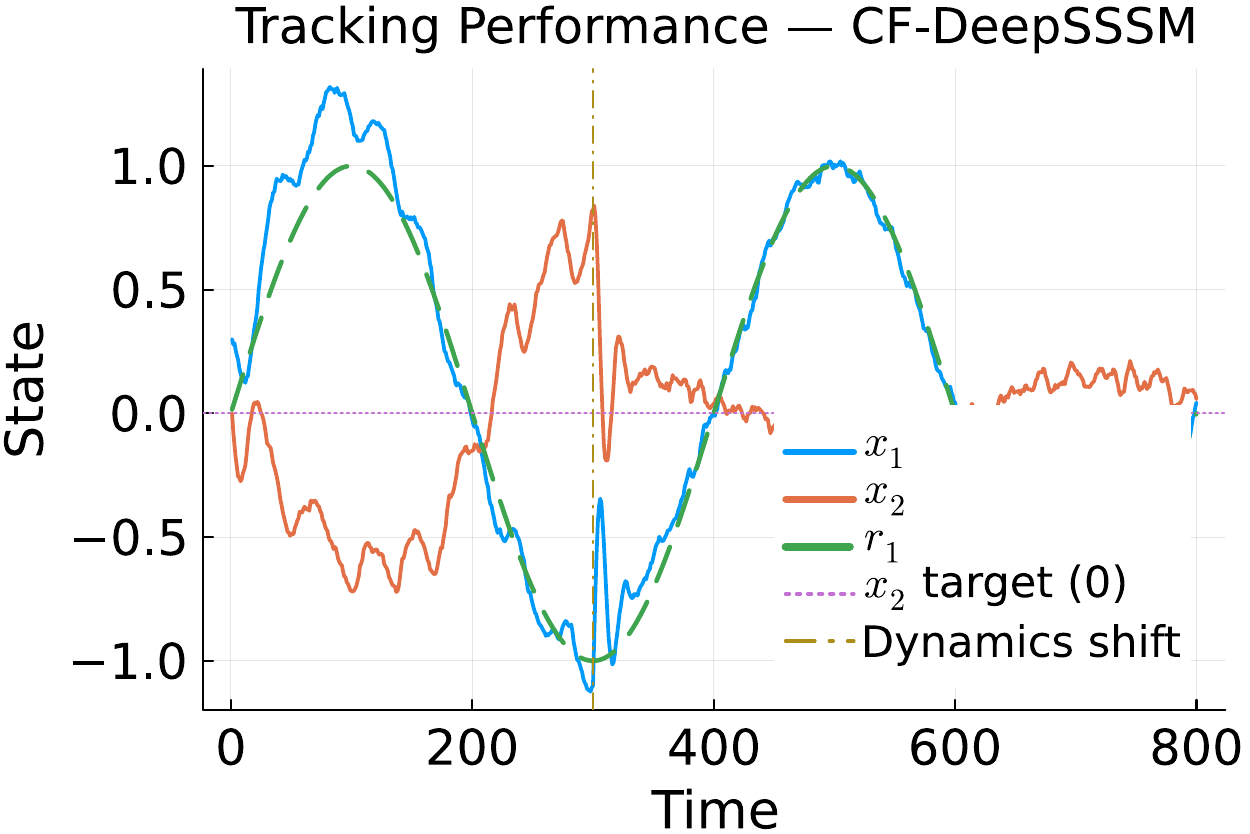}
		\caption{Tracking response.} 
\label{fig:S1_tracking}
\end{subfigure}\hfill
\begin{subfigure}[t]{0.16\textwidth}
\centering
\includegraphics[width=\linewidth]{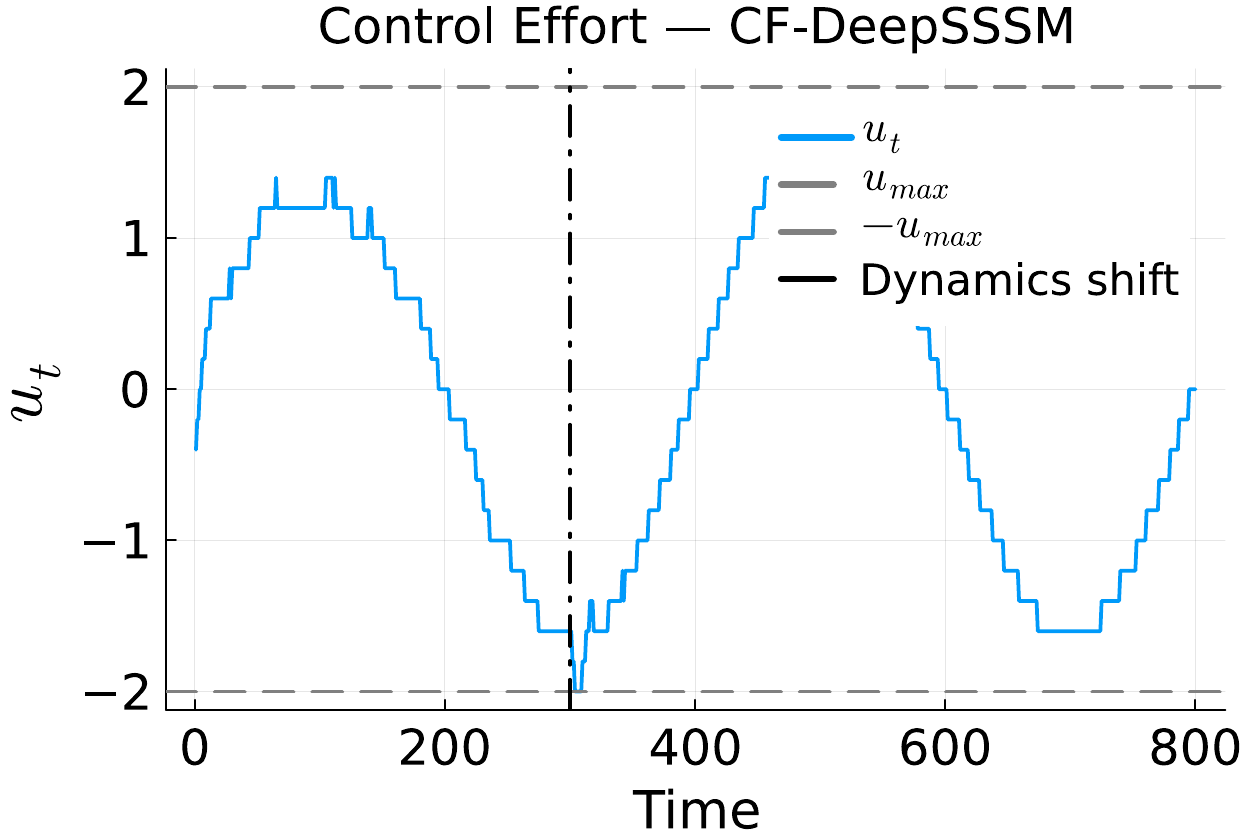}
\caption{Control input.} 
\label{fig:S1_control}
\end{subfigure}\hfill
\begin{subfigure}[t]{0.16\textwidth}
\centering
\includegraphics[width=\linewidth]{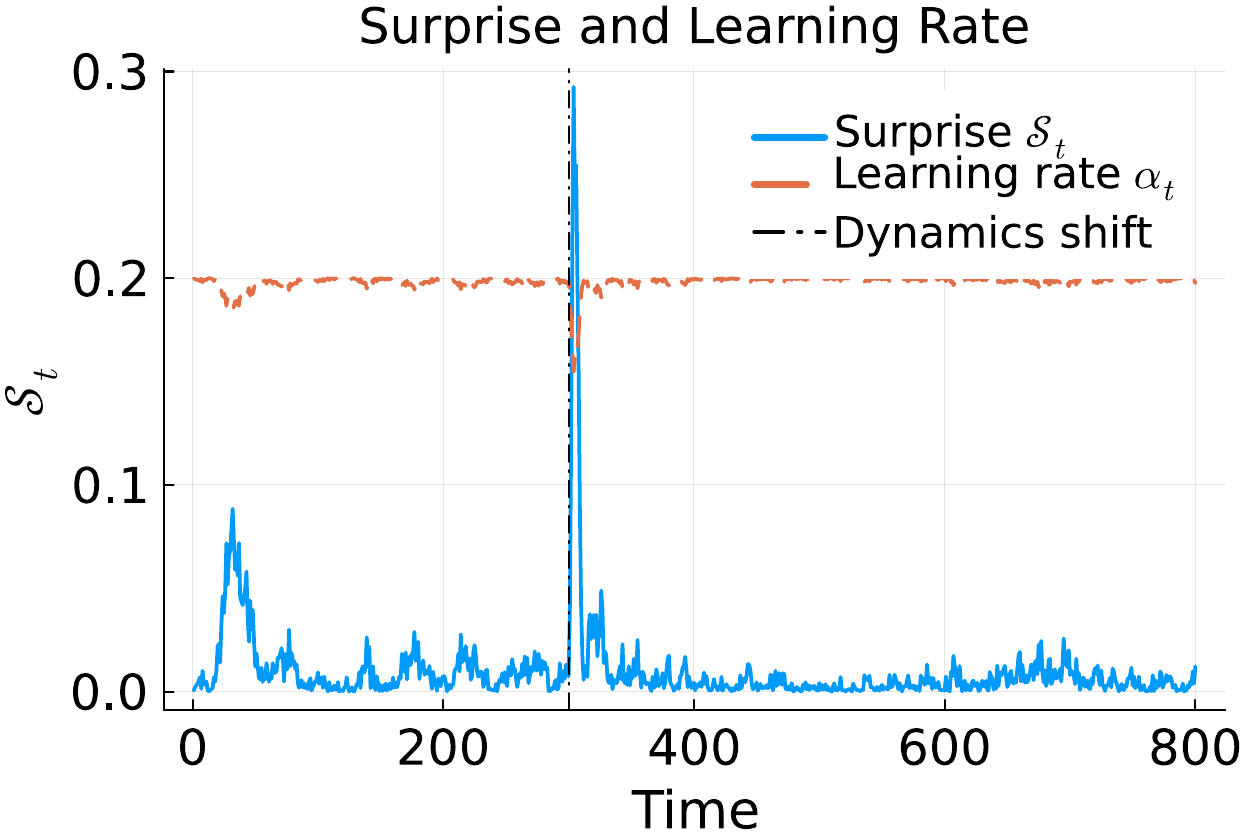}
\caption{$\mathcal{S}_t$ and $\alpha_t$.} 
\label{fig:S1_surprise}
\end{subfigure}\hfill
\begin{subfigure}[t]{0.16\textwidth}
\centering
\includegraphics[width=\linewidth]{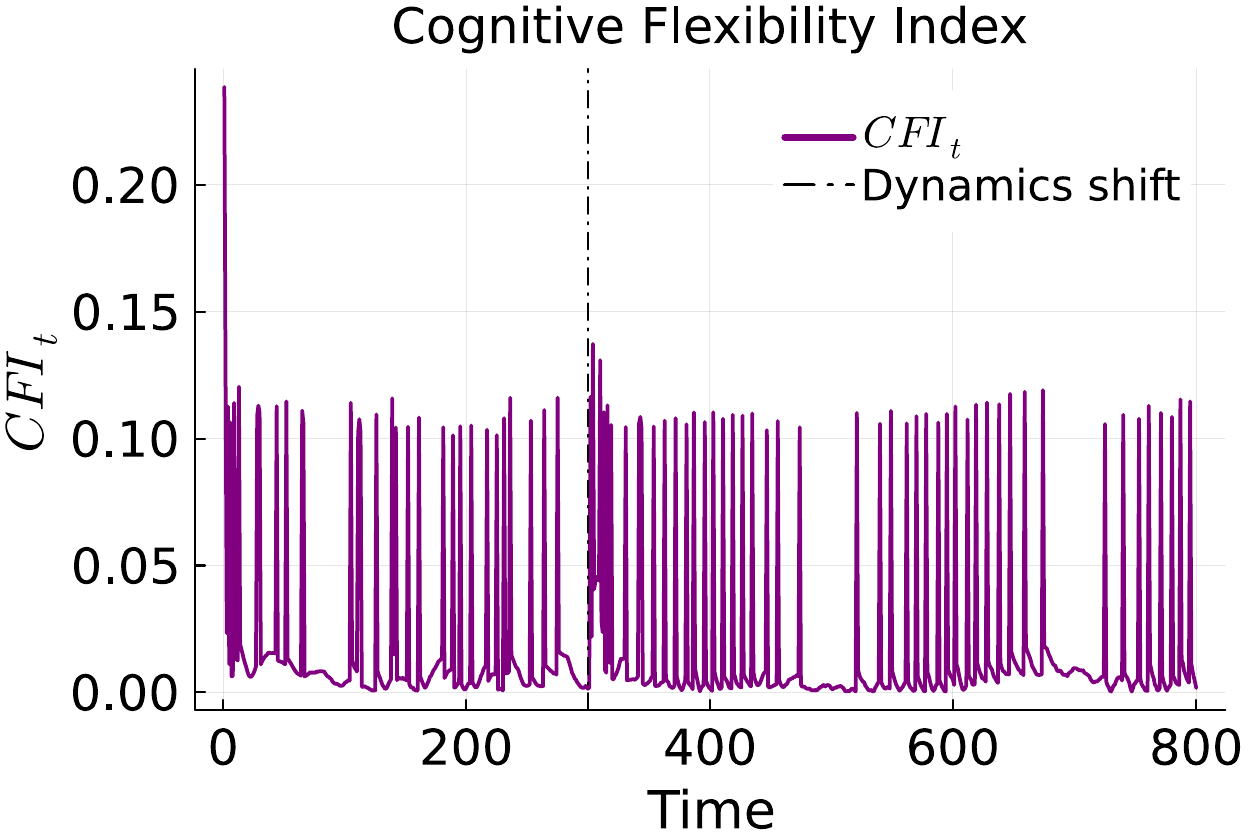}
\caption{CFI$_t$.} 
\label{fig:S1_cfi}
\end{subfigure}\hfill
\begin{subfigure}[t]{0.16\textwidth}
\centering
\includegraphics[width=\linewidth]{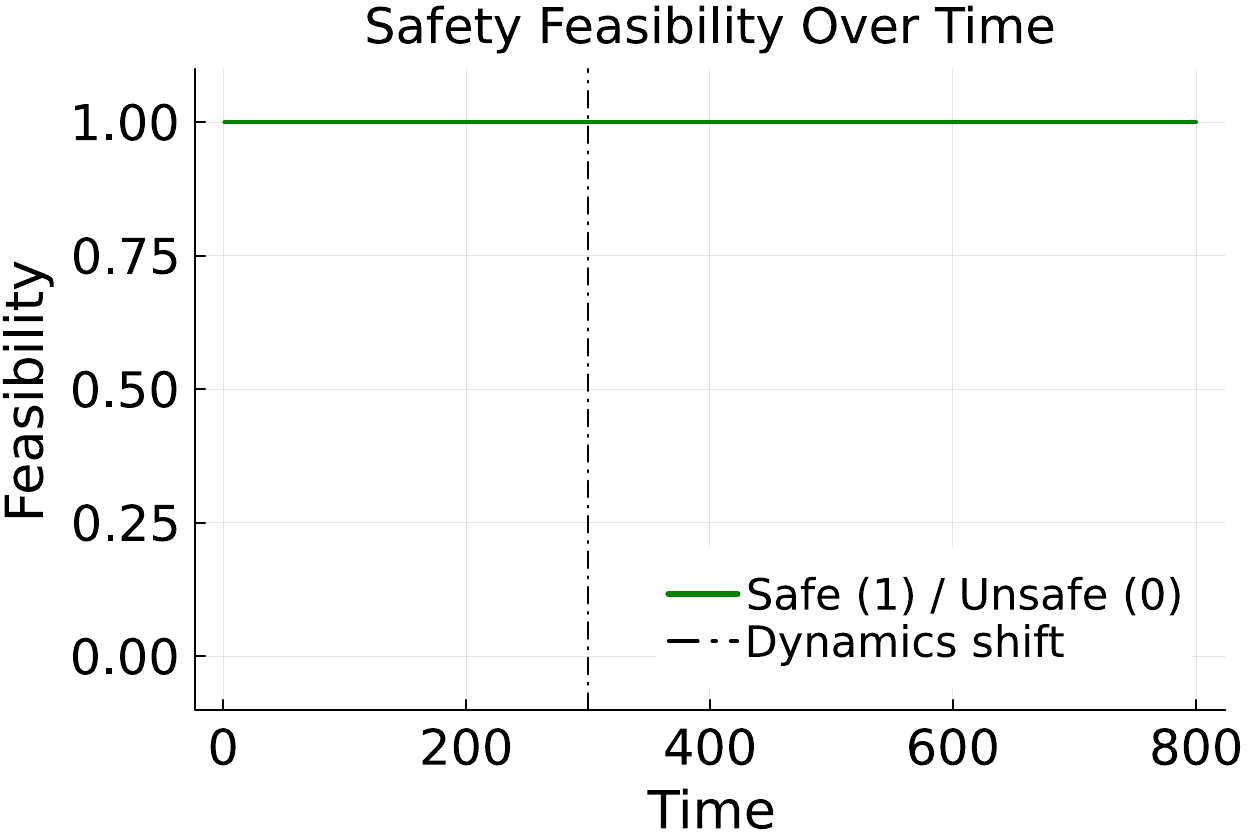}
\caption{Safety.} 
\label{fig:S1_safety}
\end{subfigure}\hfill
\begin{subfigure}[t]{0.16\textwidth}
\centering
\includegraphics[width=\linewidth]{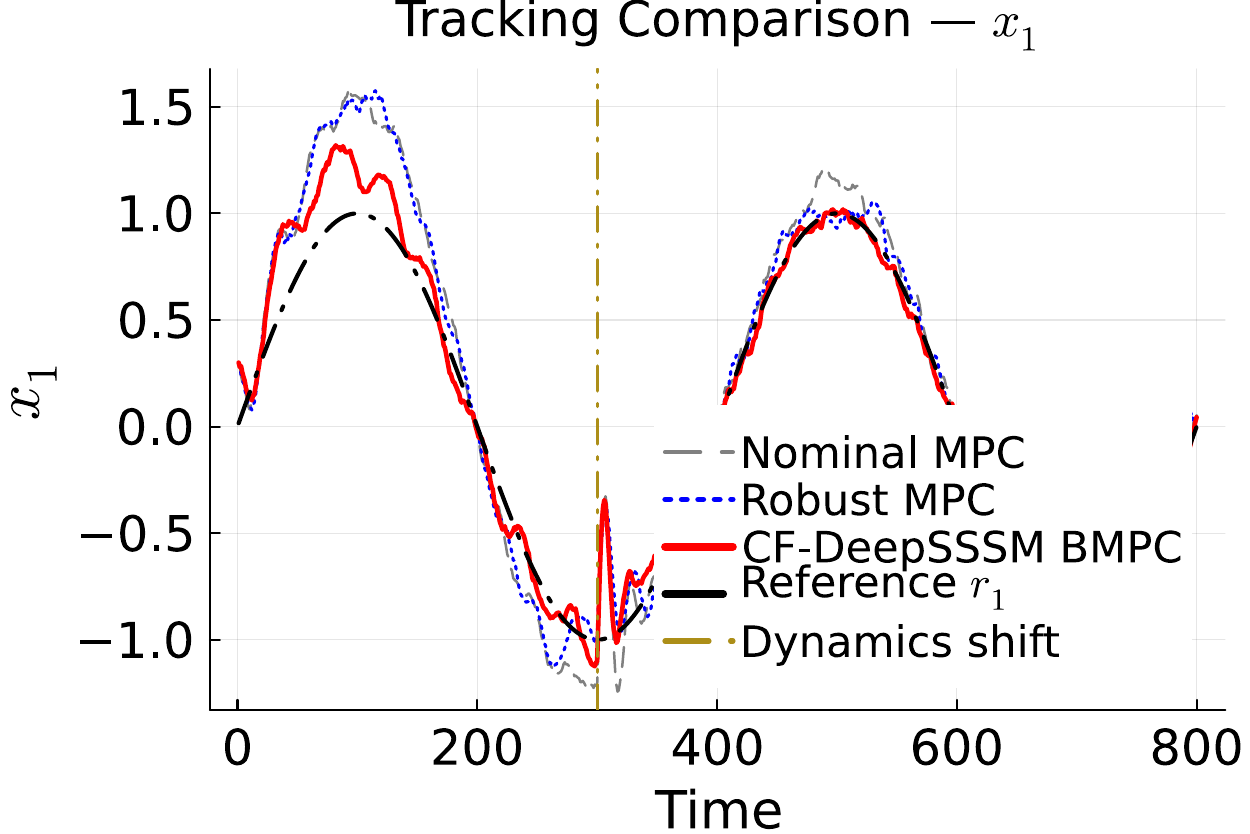}
\caption{Comparison.} 
\label{fig:S1_compare}
\end{subfigure}

\caption{Scenario~\ref{sec:scenario1}—Closed-loop response to an abrupt
dynamics shift at $t{=}300$ (vertical dashed line).
}
\label{fig:S1_overview}
\end{figure*}

\section{Simulation Studies}
\label{sec:simulation}

We validate the proposed CF--DeepSSSM BMPC controller on a
nonlinear, partially observed system with state $x_t \in \mathbb{R}^2$:
$
x_{t+1} =
\mathscr{A}_{\mathrm{env}}(t)x_t
+ \mathscr{B}_{\mathrm{env}}u_t
+ \Delta(x_t)
+ \omega_t,
\qquad
y_t =
\mathscr{C}_{\mathrm{env}}(t)x_t
+ \nu_t,
$
where $u_t \in \mathbb{R}$ and $y_t \in \mathbb{R}^2$.
We use the bounded interaction nonlinearity
$
\Delta(x_t)=
\begin{bmatrix}
0.05\tanh(x_{1,t})\\
0
\end{bmatrix}.
$
Process and measurement disturbances are zero-mean Gaussian,
$
\omega_t \sim \mathcal{N}(0,\sigma_w^2 I_2),
\qquad
\nu_t \sim \mathcal{N}(0,\sigma_v^2 I_2),
$
with $\sigma_w = 0.01$ and $\sigma_v = 0.02$.
Unless otherwise specified,
$\mathscr{A}_1 =
\begin{bmatrix}
0.97 & 0.08\\
-0.12 & 0.96
\end{bmatrix},
\quad
\mathscr{A}_2 =
\begin{bmatrix}
0.90 & 0.24\\
-0.18 & 0.91
\end{bmatrix},
$
$\mathscr{B}_{\mathrm{env}} =
\begin{bmatrix}
0.05\\
0.10
\end{bmatrix},
\quad
\mathscr{C}_0 = I_2,
$
where $\mathscr{A}_1$ is Schur stable and $\mathscr{A}_2$
remains Schur stable but exhibits stronger cross-state coupling,
inducing a structured model mismatch. 
We set $t_s = 300$, horizon $N = 10$,
cost $Q=\mathrm{diag}(8,3)$ and $R=0.05$,
and enforce CF with $\eta_{\max}=0.15$.
The task requires $x_1$ to track a sinusoidal reference
(amplitude $1$, frequency $0.02\pi$), while $x_2$ is regulated to zero,
subject to $|x_1|\le3$, $|x_2|\le3$, and $|u_t|\le2$. 
Starting from $\theta_0$, CF--DeepSSSM updates
$(\mathscr{A}_t,\mathscr{B}_t,\mathscr{C}_t)$ via bounded surprise-driven steps
(Sec.~\ref{sec:theory}; Theorem~\ref{thm:drift}). The simulations evaluate the
resulting closed-loop behavior under three representative model mismatches:
abrupt dynamics shifts, observation drift, and gradual dynamics evolution.

\subsection{Scenario \ref{sec:scenario1} — Abrupt Dynamics Shift}
\label{sec:scenario1}

At $t=t_s$, the environment dynamics switch abruptly
$\mathscr{A}_{\mathrm{env}}(t): \mathscr{A}_1 \rightarrow \mathscr{A}_2$,
while the observation mapping remains nominal,
$\mathscr{C}_{\mathrm{env}}(t)=\mathscr{C}_0$ for all $t$.
This models a sudden variation in actuator or contact dynamics,
isolating a pure dynamics-level distributional shift under reliable sensing.

\noindent\textit{Results and discussion.}
Figure~\ref{fig:S1_overview} summarizes the closed-loop response to an abrupt dynamics switch at $t=t_s=300$.
Before the switch, the controller achieves accurate tracking and stable regulation
under the nominal dynamics $\mathscr{A}_1$ (Fig.~\ref{fig:S1_tracking}).
Immediately after $\mathscr{A}_1\!\to\!\mathscr{A}_2$, a transient in~\ref{fig:S1_tracking} degradation appears due to
structural mismatch between  
$\mathscr{A}_{\mathrm{env}}(t)$ and 
$p_{\theta}(z_{t+1}\mid z_t,u_t)$ in~\eqref{eq:latentbelief},
which is reflected by a sharp increase in the predictive surprise $S_t$ (Fig.~\ref{fig:S1_surprise}),
thereby activating the online update in~\eqref{eq:param_update} with step size $\eta_t$.
By Theorem~\ref{thm:drift}, the parameter variation remains uniformly
bounded, as reflected by the localized burst and decay of $CFI_t$
(Fig.~\ref{fig:S1_cfi}), indicating incremental latent reorganization.
Despite this adaptation phase, the applied input remains within the actuation bounds
(Fig.~\ref{fig:S1_control}), and state--input feasibility is preserved for all $t$
(Fig.~\ref{fig:S1_safety}), consistent with Lemma~\ref{lem:tightening} and recursive
feasibility in Theorem~\ref{thm:feasibility} (hence Corollary~\ref{cor:safety}).
Moreover, the closed-loop belief dynamics remain ISS with respect to
bounded modeling error and bounded parameter drift, as established in
Theorem~\ref{thm:iss}, which is consistent with the bounded tracking
response in Fig.~\ref{fig:S1_tracking}. Figure~\ref{fig:S1_compare}
further compares the controllers: nominal MPC exhibits persistent
tracking bias under the fixed model, while robust MPC maintains
feasibility through fixed tightening but remains conservative.
In contrast, CF-DeepSSSM restores tracking performance through
\emph{surprise-driven representation adaptation} while preserving
predictive safety throughout the transition.


\subsection{Scenario \ref{sec:scenario2} — Observation Drift}
\label{sec:scenario2}

This scenario isolates \emph{latent representation reorganization}
under sensing degradation.
The physical dynamics remain fixed,
$\mathscr{A}_{\mathrm{env}}(t)=\mathscr{A}_1$ for all $t$,
while the observation mapping drifts after $t_s$.
Specifically, we set
$\mathscr{C}_{\mathrm{env}}(t)=\mathscr{C}_0$ for $t<t_s$,
and for $t\ge t_s$ apply a smooth \emph{observation drift}
$
\mathscr{C}_{\mathrm{env}}(t)
=
\mathscr{C}_0
+
\kappa(t)\,\Delta \mathscr{C},
\qquad
\kappa(t)
=
1-\exp\!\left(-\frac{t-t_s}{\tau_c}\right).
$
Unless otherwise stated, we use
$
\Delta \mathscr{C}
=
\begin{bmatrix}
0.35 & 0.10\\
-0.10 & 0.25
\end{bmatrix},
\qquad
\tau_c=80.
$
Such observation drift may arise from sensor miscalibration,
viewpoint changes, or tactile bias, distorting the measurement--state
mapping while leaving the physical dynamics unchanged.


\begin{figure}[t]
\centering
\begin{subfigure}[t]{0.25\textwidth}
\centering
\includegraphics[height=1.6cm]{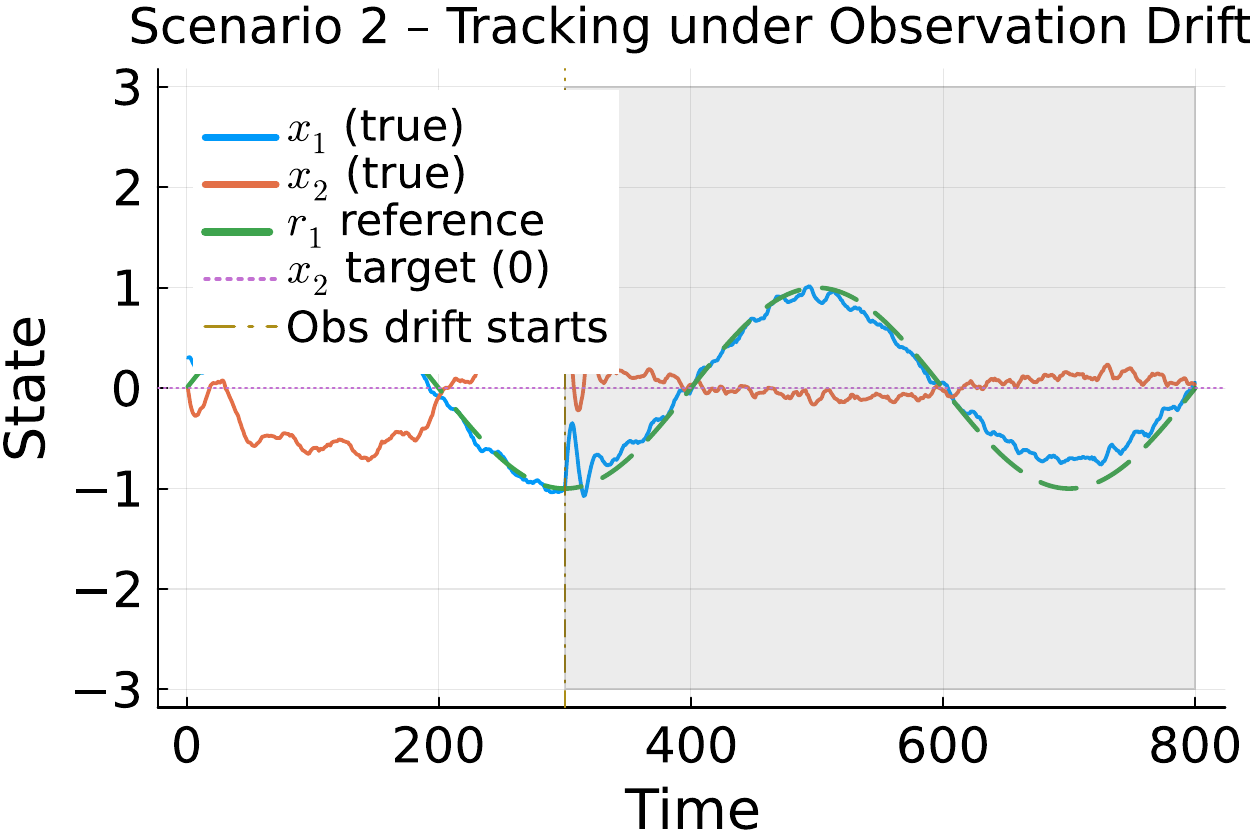}
\caption{State tracking.} 
\label{fig:S2_tracking}
\end{subfigure}
\begin{subfigure}[t]{0.25\textwidth}
\centering
\includegraphics[height=1.6cm]{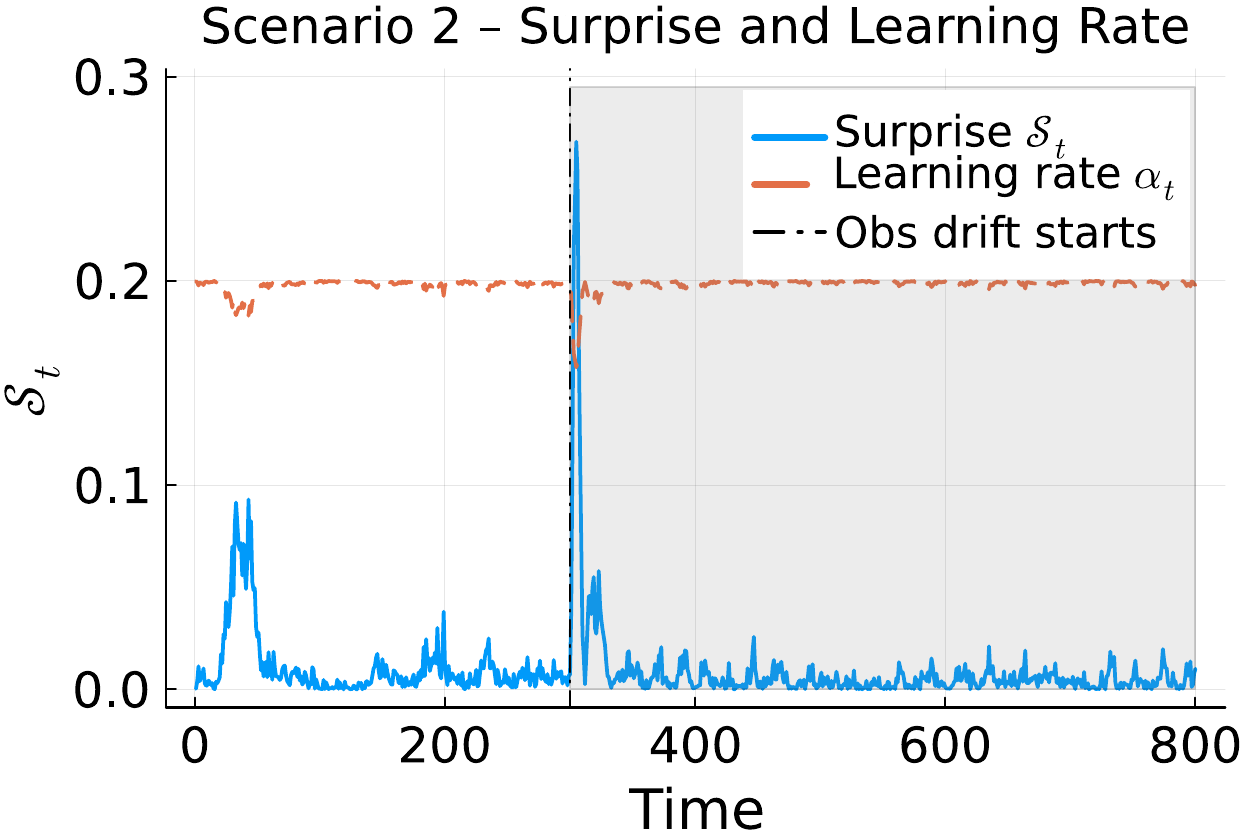}
\caption{$\mathcal{S}_t$ and bounded learning.}
\label{fig:S2_surprise}
\end{subfigure}\hfill
\begin{subfigure}[t]{0.25\textwidth}
\centering
\includegraphics[height=1.6cm]{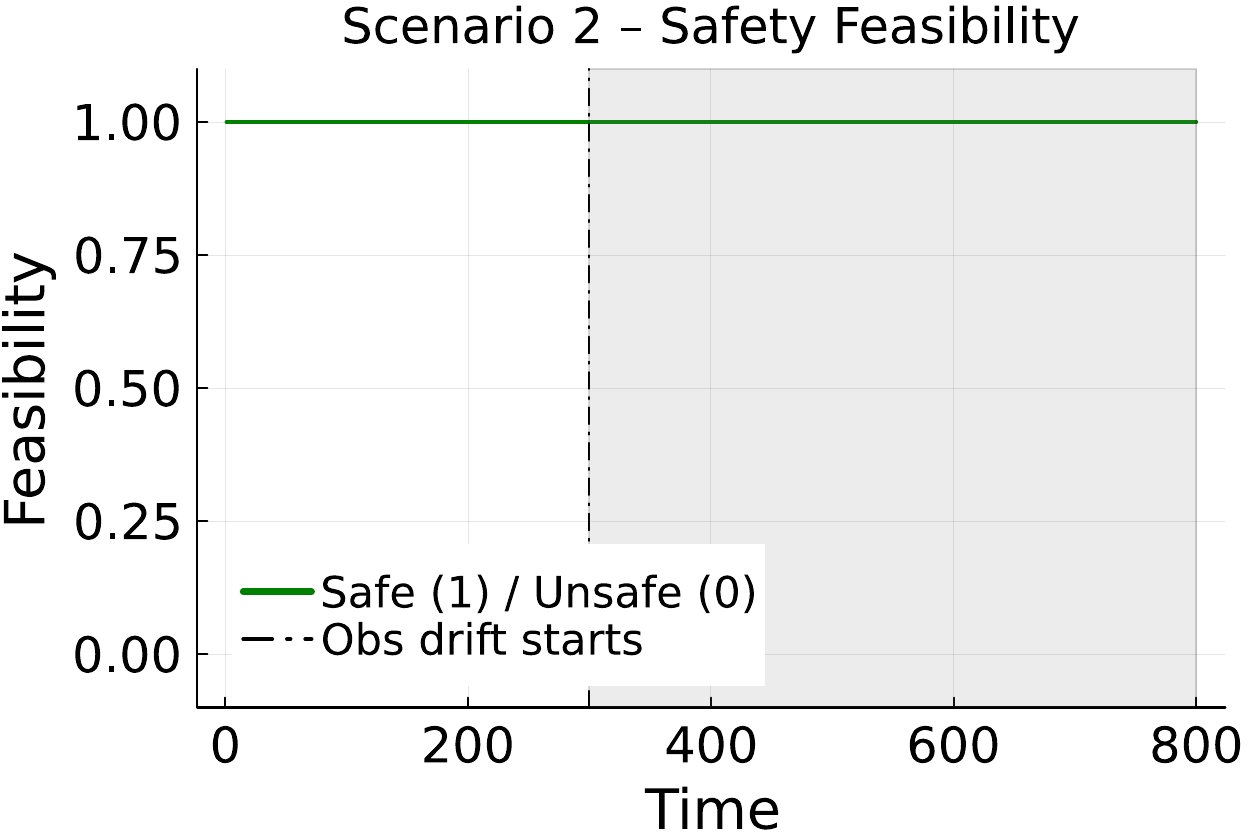}
\caption{Safety feasibility.} 
\label{fig:S2_safety}
\end{subfigure}
\begin{subfigure}[t]{0.25\textwidth}
\centering
\includegraphics[height=1.6cm]{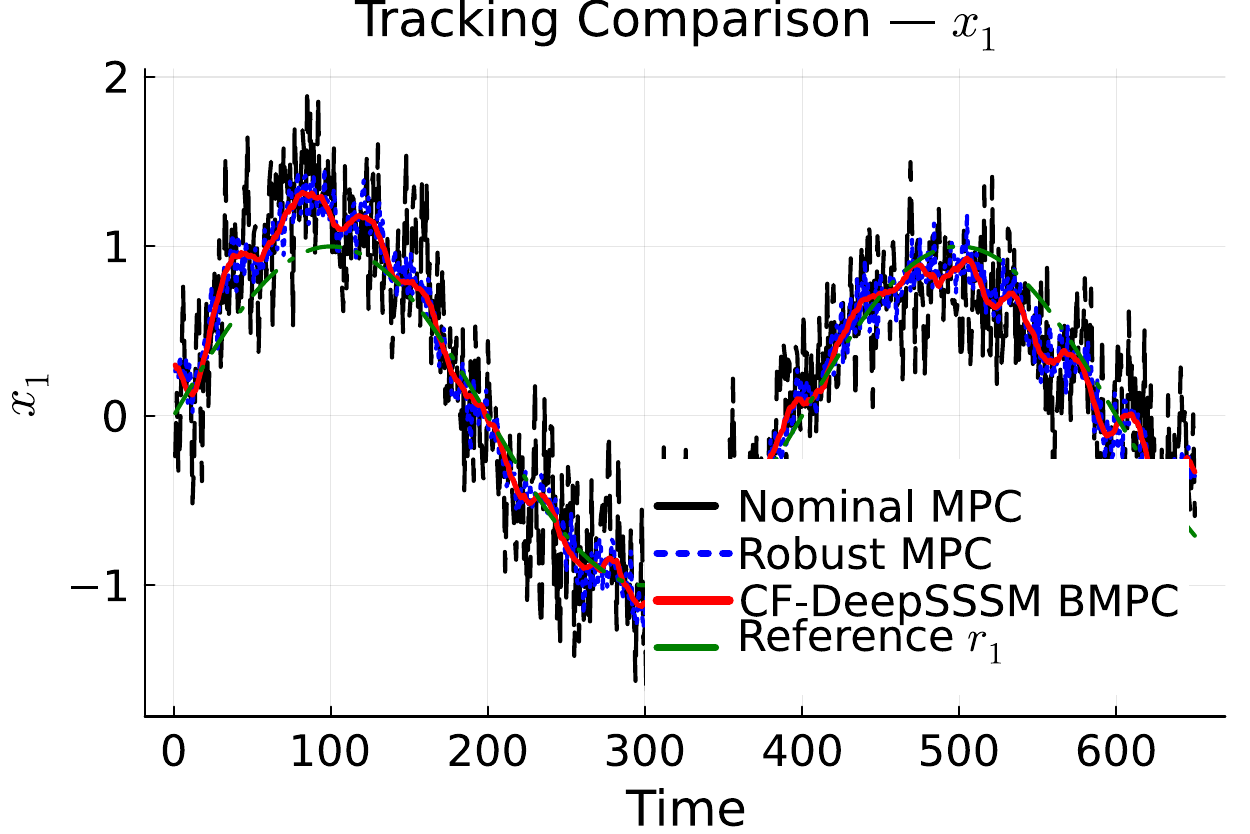}
\caption{Tracking comparison.} 
\label{fig:S2_compare}
\end{subfigure}

\caption{Scenario~\ref{sec:scenario2} --- Observation drift after $t=300$ (shaded region).}
\label{fig:scenario2_overview}
\end{figure}

\noindent\textit{Results and discussion.}
Figure~\ref{fig:scenario2_overview} evaluates the controller under
progressive observation drift starting at $t_s=300$, where the
measurement-to-state mapping deviates from the nominal model while
the physical dynamics remain unchanged.
Fig.~\ref{fig:S2_tracking} shows that CF--DeepSSSM maintains accurate
tracking despite the corrupted sensing channel. Unlike
Scenario~\ref{sec:scenario1}, where adaptation compensates for a
dynamics mismatch, the shift here arises solely from the observation
channel; adaptation therefore acts on the observation model that
maps measurements to the latent belief.
The predictive surprise $S_t$ (Fig.~\ref{fig:S2_surprise}) rises after
drift onset and triggers the bounded update in~\eqref{eq:param_update},
correcting the resulting inference bias. Fig.~\ref{fig:S2_safety}
confirms that state--input constraints remain satisfied, consistent
with uncertainty-aware tightening (Lemma~\ref{lem:tightening}) and
recursive feasibility (Theorem~\ref{thm:feasibility}).
Fig.~\ref{fig:S2_compare}, 
nominal MPC relies on the biased belief and exhibits
persistent tracking error, while robust MPC preserves feasibility
but cannot correct the perception bias. CF--DeepSSSM instead
restores tracking by adapting the observation mapping.
\noindent\textbf{Additional baselines and ablations.}
We further evaluate: (i) \emph{Fixed DeepSSSM} (no adaptation),
(ii) adaptation without cognitive-flexibility rate limiting, and
(iii) adaptation without uncertainty-aware tightening.
Results (Fig.}~\ref{fig:S2_compare}, Fig.~\ref{fig:scenario_vc_tracking})
show that removing either mechanism degrades tracking or violates
constraints, whereas the full CF--DeepSSSM maintains performance and feasibility.

\subsection{Scenario \ref{sec:scenario3} --- Gradual Dynamics Drift}
\label{sec:scenario3}

We consider a gradual drift in the dynamics matrix
$\mathscr{A}_{\mathrm{env}}(t)$ starting at $t_s$, emulating \emph{slowly
varying interaction} conditions (cf. Fig.~\ref{fig:scenario_vc_overview}).
The observation mapping remains nominal,
$\mathscr{C}_{\mathrm{env}}(t)=\mathscr{C}_0$, and
$
\mathscr{A}_{\mathrm{env}}(t)
=
\mathscr{A}_1+\kappa(t)(\mathscr{A}_2-\mathscr{A}_1),\qquad
\kappa(t)=1-\exp\!\left(-\frac{t-t_s}{\tau_a}\right),
$
with $\kappa(t)=0$ for $t<t_s$ and $\tau_a=120$ unless otherwise stated.
Such drift may arise from evolving contact properties, actuator heating,
or stiffness changes, and therefore tests whether CF--DeepSSSM can
sustain latent adaptation, feasibility, and tracking under slowly
accumulating model mismatch.

\noindent\textit{Initialization and safe early phase.}
Parameters are initialized with the offline-pretrained model $\theta_0$
obtained in the nominal (pre-drift) regime (Scenario~\ref{sec:scenario1}).
After drift onset, updates satisfy the bounded learning-rate condition
of Theorem~\ref{thm:drift}, while safety is maintained by uncertainty-aware
tightening and recursive feasibility (Lemma~\ref{lem:tightening},
Theorem~\ref{thm:feasibility}), implying state--input safety for all $t$
(Corollary~\ref{cor:safety}).

\noindent\textit{Mild nonlinear interaction.}
To emulate realistic interaction effects while preserving stability, we
augment the drift dynamics with a small bounded state-dependent perturbation,
$
x_{t+1}
=
\mathscr{A}_{\mathrm{env}}(t)x_t+\mathscr{B}_{\mathrm{env}}u_t+w_t+
\begin{bmatrix}
\rho \tanh(x_{1,t})\\
0
\end{bmatrix},
$
where $\rho>0$ is small and $w_t\sim\mathcal{N}(0,\sigma_w^2 I)$.
The term $\tanh(\cdot)$ models a bounded interaction nonlinearity
(e.g., compliance or contact saturation). Together with the gradual drift
in $\mathscr{A}_{\mathrm{env}}(t)$, it induces a persistent prediction bias
over the horizon, providing a controlled setting to evaluate
representation adaptation under sustained dynamics variation.

\textit{Results and discussion.}
Figure~\ref{fig:scenario_vc_overview} evaluates the controller under
gradual dynamics drift, where $\mathscr{A}_{\mathrm{env}}(t)$ evolves
continuously rather than switching abruptly as in Scenario~\ref{sec:scenario1}.
Fig.~\ref{fig:scenario_vc_tracking} shows that fixed-model baselines
cannot compensate for this progressive mismatch: nominal MPC exhibits
persistent tracking degradation, while robust MPC preserves feasibility
via fixed tightening but remains conservative. CF--DeepSSSM instead adapts the latent predictor online via 
~\eqref{eq:param_update},
allowing the belief dynamics induced by~\eqref{eq:latentbelief}
to track the drifting environment. Fig.~\ref{fig:scenario_vc_cfi} reports
the CFI, which quantifies the magnitude
of representation updates. Despite the continuous drift, the index
remains bounded and localized, indicating incremental \emph{belief
reorganization}. This behavior is consistent with the bounded posterior
drift guarantee in Theorem~\ref{thm:drift}.


\begin{figure}[t]
\centering

\begin{subfigure}[t]{0.48\linewidth}
\centering
\includegraphics[width=\linewidth,height=1.9cm,keepaspectratio]
{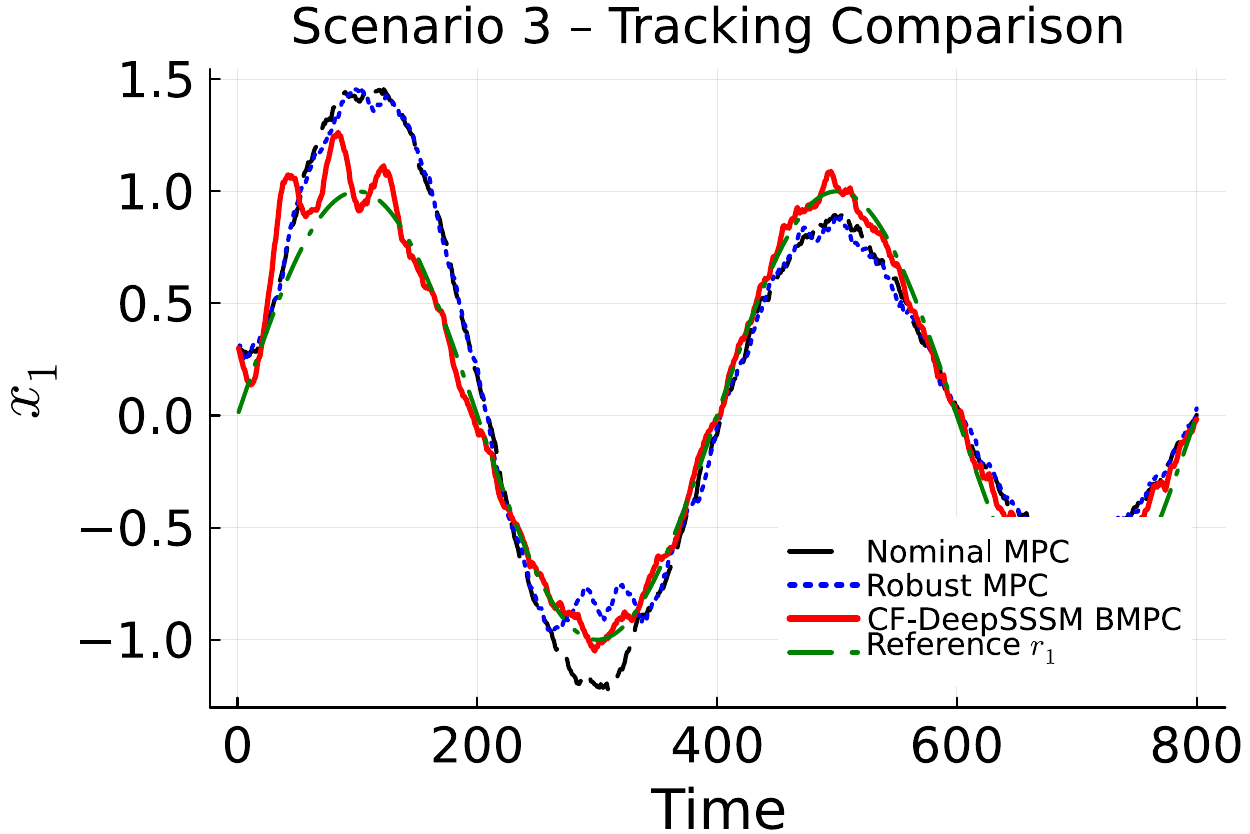}
\caption{Baseline comparison.} 
\label{fig:scenario_vc_tracking}
\end{subfigure}
\hfill
\begin{subfigure}[t]{0.48\linewidth}
\centering
\includegraphics[width=\linewidth,height=1.9cm,keepaspectratio]
{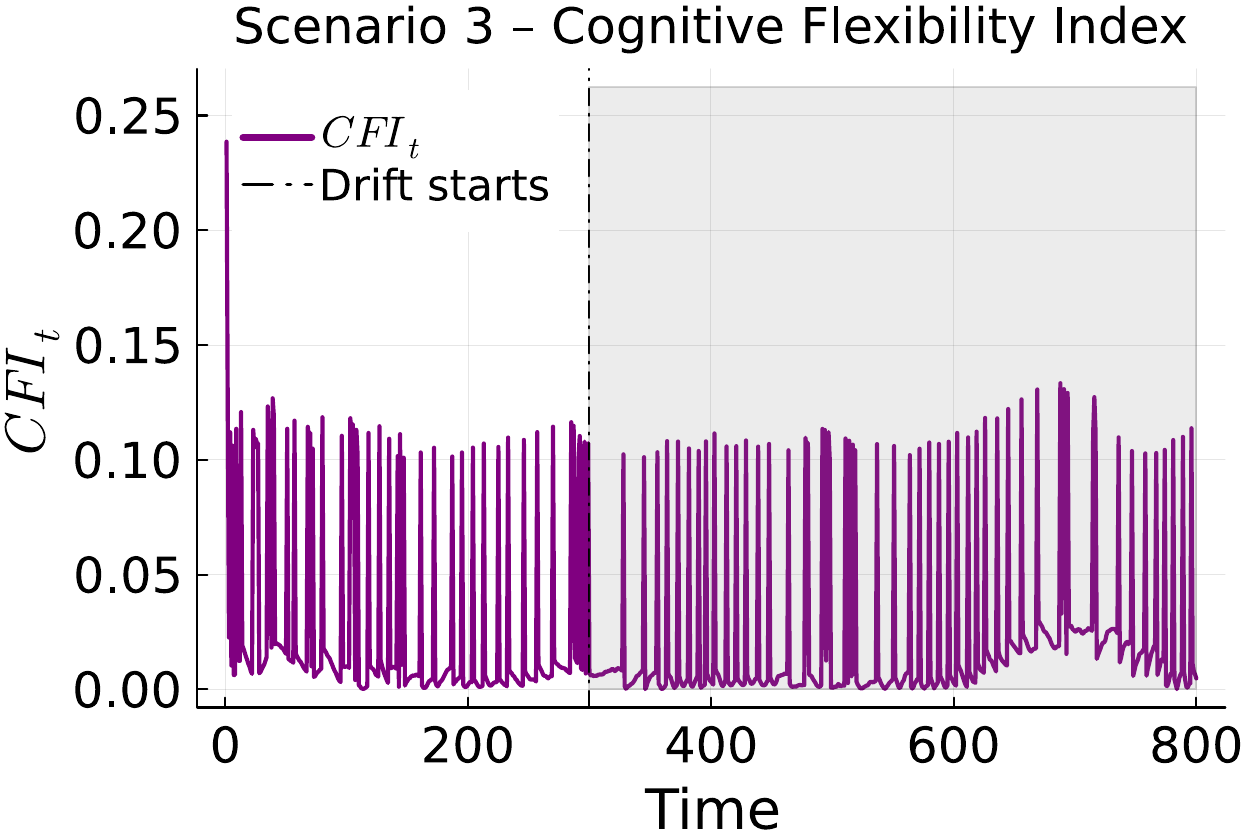}
\caption{CFI.}
\label{fig:scenario_vc_cfi}
\end{subfigure}

\caption{Scenario~\ref{sec:scenario3}. 
Gradual dynamics drift.} 
\label{fig:scenario_vc_overview}
\end{figure}


\section{Conclusion}
\label{sec:conclusion}

This letter presented \emph{CF--DeepSSSM}, a framework for safety-critical
control under partial observability and distributional shift. The approach
integrates uncertainty-aware latent dynamics learning, surprise-regulated model
adaptation, and BMPC with probabilistic safety constraints in a unified
closed loop. The key idea is \emph{regulated latent reorganization}: internal
representations adapt to predictive mismatch while their evolution is
explicitly bounded to preserve stability and safety. We established guarantees
on bounded posterior drift, recursive feasibility, and ISS of the closed-loop
belief dynamics, and validated them in simulation under abrupt dynamics and
observation shifts as well as gradual dynamics drift. 

\bibliographystyle{IEEEtran}
\bibliography{refs}

@article{Derler2012CPS,
  author={Derler, P. and Lee, E. A. and Sangiovanni-Vincentelli, A.},
  title={Modeling cyber--physical systems},
  journal={Proceedings of the IEEE},
  volume={100},
  number={1},
  pages={13--28},
  year={2012},
  month={Jan}
}

@article{Hewing2020LearningMPC,
  author={Hewing, L. and Wabersich, K. P. and Menner, M. and Zeilinger, M. N.},
  title={Learning-Based Model Predictive Control: Toward Safe Learning in Control},
  journal={Annual Review of Control, Robotics, and Autonomous Systems},
  volume={3},
  number={1},
  pages={269--296},
  year={2020}
}

@article{Becker2021,
  author={Wabersich, K. P. and Hewing, L. and Carron, A. and Zeilinger, M. N.},
  title={Probabilistic Model Predictive Safety Certification for Learning-Based Control},
  journal={IEEE Transactions on Automatic Control},
  volume={67},
  number={1},
  pages={176--188},
  year={2022},
  month={Jan}
}

@article{Gedon2021DeepSSSM,
  author={Gedon, D. and Wahlström, N. and Schön, T. B. and Ljung, L.},
  title={Deep state-space models for nonlinear system identification},
  journal={IFAC-PapersOnLine},
  volume={54},
  number={7},
  pages={481--486},
  year={2021}
}

@inproceedings{Fraccaro2017,
  author={Fraccaro, M. and Sønderby, S. K. and Paquet, U. and Winther, O.},
  title={Sequential Neural Models with Stochastic Layers},
  booktitle={Advances in Neural Information Processing Systems (NeurIPS)},
  year={2017}
}

@inproceedings{Karl2017,
  author={Karl, M. and Soelch, M. S. and Bayer, J. and van der Smagt, P.},
  title={Deep Variational Bayes Filters: Unsupervised Learning of State Space Models from Raw Data},
  booktitle={Proc. International Conference on Learning Representations (ICLR)},
  year={2017}
}

@inproceedings{Hafner2019,
  author={Hafner, D. and Lillicrap, T. and Fischer, I. and Villegas, R. and Ha, D. and Lee, H. and Davidson, J.},
  title={Learning Latent Dynamics for Planning from Pixels},
  booktitle={Proc. International Conference on Machine Learning (ICML)},
  year={2019}
}

@book{Slotine1991,
  author={Slotine, J.-J. E. and Li, W.},
  title={Applied Nonlinear Control},
  publisher={Prentice Hall},
  year={1991}
}

@book{Ioannou1996,
  author={Ioannou, P. A. and Sun, J.},
  title={Robust Adaptive Control},
  publisher={Prentice Hall},
  year={1996}
}

@article{Aswani2013,
  author={Aswani, A. and Gonzalez, H. and Sastry, S. S. and Tomlin, C.},
  title={Provably Safe and Robust Learning-Based Model Predictive Control},
  journal={Automatica},
  volume={49},
  number={5},
  pages={1216--1226},
  year={2013}
}

@article{Wabersich2021PSF,
  author={Wabersich, K. P. and Zeilinger, M. N.},
  title={Predictive Control Barrier Functions: Enhanced Safety Mechanisms for Learning-Based Control},
  journal={IEEE Transactions on Automatic Control},
  volume={68},
  number={5},
  pages={2638--2651},
  year={2023},
  month={May}
}

@article{Soloperto2023BayesianActuation,
  author={Soloperto, R. and M{\"u}ller, M. A. and Trimpe, S. and Allg{\"o}wer, F.},
  title={Learning-Based Robust Model Predictive Control with State-Dependent Uncertainty},
  journal={IFAC-PapersOnLine},
  volume={51},
  number={20},
  pages={442--447},
  year={2018}
}

@inproceedings{Achiam2017,
  author={Achiam, J. and Held, D. and Tamar, A. and Abbeel, P.},
  title={Constrained Policy Optimization},
  booktitle={Proc. International Conference on Machine Learning (ICML)},
  year={2017}
}

@article{Chow2018,
  author={Chow, Y. and Ghavamzadeh, M. and Janson, L. and Pavone, M.},
  title={Risk-Constrained Reinforcement Learning with Percentile Risk Criteria},
  journal={Journal of Machine Learning Research},
  volume={18},
  number={167},
  pages={1--51},
  year={2018}
}

@article{Berkenkamp2021SafeRL,
  author={Brunke, L. and Greeff, M. and Hall, A. W. and Yuan, Z. and Zhou, S. and Panerati, J. and Schoellig, A. P.},
  title={Safe Learning in Robotics: From Learning-Based Control to Safe Reinforcement Learning},
  journal={Annual Review of Control, Robotics, and Autonomous Systems},
  volume={5},
  pages={411--444},
  year={2022}
}

@article{Thananjeyan2021Safety,
  author={Thananjeyan, B. and Balakrishna, A. and Rosolia, U. and Li, F. and McAllister, R. and Gonzalez, J. E. and Levine, S. and Borrelli, F. and Goldberg, K.},
  title={Safety Augmented Value Estimation From Demonstrations {(SAVED)}: Safe Deep Model-Based RL for Sparse Cost Robotic Tasks},
  journal={IEEE Robotics and Automation Letters},
  volume={5},
  number={2},
  pages={3612--3619},
  year={2020}
}

@article{Scott1962,
  author={Scott, W. A.},
  title={Cognitive Complexity and Cognitive Flexibility},
  journal={Sociometry},
  volume={25},
  number={4},
  pages={405--414},
  year={1962}
}

@article{BelmonteBaeza2022MetaRL,
  author={Belmonte-Baeza, A. and Lee, J. and Valsecchi, G. and Hutter, M.},
  title={Meta Reinforcement Learning for Optimal Design of Legged Robots},
  journal={IEEE Robotics and Automation Letters},
  volume={7},
  number={4},
  pages={12134--12141},
  year={2022}
}

@article{McClement2022MetaRL,
  author={McClement, D. G. and Lawrence, N. P. and Forbes, M. G. and Loewen, P. D. and Backström, J. U. and Gopaluni, R. B.},
  title={Meta-Reinforcement Learning for Adaptive Control of Second Order Systems},
  journal={arXiv preprint arXiv:2209.09301},
  year={2022}
}

@inproceedings{Goldie2024LearnedOptRL,
  author={Goldie, A. D. and Lu, C. and Jackson, M. T. and Whiteson, S. and Foerster, J. N.},
  title={Can Learned Optimization Make Reinforcement Learning Less Difficult?},
  booktitle={Advances in Neural Information Processing Systems (NeurIPS)},
  year={2024}
}

@inproceedings{Goldie2025MetaRL,
  author={Goldie, A. D. and Wang, Z. and Cohen, J. and Foerster, J. N. and Whiteson, S.},
  title={How Should We Meta-Learn Reinforcement Learning Algorithms?},
  booktitle={Reinforcement Learning Conference (RLC)},
  note={Also available as arXiv:2507.17668},
  year={2025}
}

\end{document}